\documentclass[aps,prl,twocolumn,showpacs,final,nofootinbib,floatfix]{revtex4}

\usepackage{amsmath,amsfonts,amsthm,amssymb}
\usepackage{graphicx}
\usepackage{braket}
\usepackage{setspace}
\usepackage{bbm}

\DeclareMathOperator{\tr}{tr}
\DeclareMathOperator{\Prob}{Prob}
\newcommand\GHZ{{\mathrm{GHZ}}}
\newtheorem*{proposition*}{Proposition}

\begin{document}


\title{Systematic errors in current quantum state tomography tools}
\author{Christian~Schwemmer,$^{1,2}$~Lukas~Knips,$^{1,2}$~Daniel~Richart,$^{1,2}$~and~Harald~Weinfurter$^{1,2}$}
\affiliation{$^{\it 1}$Max-Planck-Institut f\"ur Quantenoptik, Hans-Kopfermann-Str. 1, D-85748 Garching, Germany}
\affiliation{$^{\it 2}$Department f\"ur Physik, Ludwig-Maximilians-Universit\"at, D-80797 M\"unchen, Germany}

\author{Tobias~Moroder,$^{3}$~Matthias~Kleinmann,$^{3}$~and~Otfried~G\"uhne,$^{3}$}
\affiliation{$^{\it 3}$Naturwissenschaftlich-Technische Fakult\"at, Universit\"at Siegen, Walter-Flex-Str.~3, D-57068 Siegen, Germany}

\begin{abstract}
Common tools for obtaining physical density matrices in experimental quantum state tomography are shown here to cause systematic errors. 
For example, using maximum likelihood or least squares optimization for state reconstruction, we observe a systematic 
underestimation of the fidelity and an overestimation of entanglement. A solution for this problem can be achieved by a linear 
evaluation of the data yielding reliable and computational simple bounds including error bars.
\end{abstract}


\maketitle

%
%
\textit{Introduction.}---Quantum state tomography (QST)~\cite{qse_book} enables us to fully determine the state of a quantum 
system and thereby to deduce all its properties. As such QST is widely used to characterize and to evaluate numerous experimentally 
implemented qubit states or their dynamics, e.g., in ion trap experiments~\cite{haeffner05a,home09}, photonic 
systems~\cite{james01a,resch05}, superconducting circuits \cite{dicarlo09}, or nuclear
magnetic resonance systems \cite{mangold04,chuang97}. The increasing complexity of today’s multiqubit/qudit quantum 
systems brought new challenges but also progress. Now, highly efficient methods allow an even scalable analysis for important subclasses of 
states~\cite{gross10a, moroder12a}. The calculation of errors of QST was significantly improved although the errors 
remain numerically expensive to evaluate for larger systems~\cite{christandl11a}.
Moreover QST was used to detect systematic errors in the alignment of an experiment itself~\cite{moroder13a}.

A central step in QST is to establish the state from the acquired experimental data. A direct, linear evaluation of the data 
returns almost for sure an unphysical density matrix with negative eigenvalues~\cite{smithey93}. 
Thus, several schemes have been developed to obtain a physical state which resembles the observed data as closely as possible 
\cite{hradil97a, james01a,blume06a}.

\begin{figure}
\includegraphics[width=\linewidth]{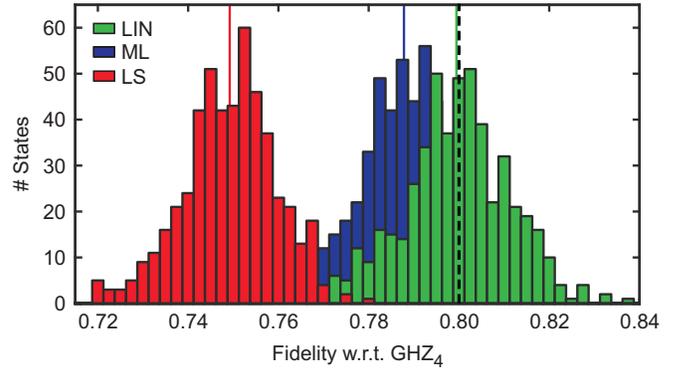}
\caption{(Color online) Histogram of the fidelity estimates of 500 independent simulations of QST of a noisy four-party Greenberger-Horne-Zeilinger (GHZ) state
for three different reconstruction schemes. The values obtained via maximum likelihood (ML, blue) or least squares (LS, red) fluctuate around a value that is lower than the initial fidelity of $80\%$ (dashed line). For comparison, we also show the result using linear inversion (LIN, green), which does not suffer from such a systematic error called bias. 
}
\label{fig:sample1}
\end{figure}

In this Letter we test whether the na\"ive expectation is met that QST delivers proper estimates for physical quantities.
We test this for the two most commonly used reconstruction schemes---maximum likelihood (ML)~\cite{hradil97a} and least squares 
(LS)~\cite{james01a}---using Monte Carlo simulations. This expectation is not fulfilled: 
both schemes return states which deviate systematically from the true state, e.g., underestimate the fidelity as shown in Fig.~\ref{fig:sample1}. For data sizes typical in multiqubit experiments the deviation from the true value is significant, in fact it is larger than commonly deduced ``error bars'' 
\cite{bootstrapping}. We show 
that the constraint of physicality necessarily leads to systematic errors for the reconstruction scheme. The size of these errors depends on the experimental noise and unavoidable statistical fluctuations. We find that it is advisable to evaluate linear operators directly on the raw data.
We also show how physical quantities that are given by convex (concave) nonlinear functions of the density matrix like the bipartite 
negativity etc., can be linearized thereby providing a meaningful lower (upper) bound, namely a directly computable error bar.

%
%
\textit{Standard state tomography tools.}---The aim of QST is to identify the initially unknown state 
$\varrho_0$ of a system via appropriate measurements on multiple preparations of this state. For an $n$-qubit system, the so-called 
Pauli tomography scheme consists of measuring in the eigenbases of all $3^n$ possible combinations of local Pauli operators, 
each yielding $2^n$ possible results \cite{james01a}. In more general terms, in a tomography protocol one repeats for each 
measurement setting~$s$ the experiment a certain number of times $N_s$ and obtains $c^s_{r}$ times the result $r$. These numbers 
then yield the frequencies $f^s_{r}=c^s_{r}/N_s$. The probability to observe the outcome $r$ for setting $s$ is given by 
$P_{\varrho_0}^s(r)=\tr(\varrho_0 M^s_{r})$. Here, $M^s_{r}$ labels the measurement operator corresponding to 
the result $r$ when measuring  setting~$s$. The probabilities $P_{\varrho_0}^s(r)$ will uniquely identify the 
unknown state $\varrho_0$, if the set of operators $M^s_{r}$ spans the space of Hermitian operators.

Provided the data $f$, i.e., the set of experimentally determined frequencies $f_r^s$ one requires a 
method to determine the estimate $\hat \varrho\equiv \hat \varrho(f)$ of the unknown state $\varrho_0$. 
Simply inverting the relations for $P_{\varrho_0}^s(r)$ we obtain
\begin{equation}
 \hat\varrho_{\rm LIN} = \sum_{r,s} A_r^s f_r^s
 \label{eq:rholin}
\end{equation}
where $A_r^s$ are determined from the measurement operators $M_r^s$ \cite{chuang97,poyatos97}.
Note that there is a canonical construction of $A_r^s$ even for the case of an overcomplete set of $M_r^s$, see SM\,1.
This reconstruction of $\hat \varrho_{\rm LIN}$ is computationally simple and has become known as linear inversion (LIN).

Yet, due to unavoidable statistical fluctuations the estimate $\hat \varrho_{\rm LIN}$ is not a physical density operator for typical experimental situations, i.e., 
generally some eigenvalues are negative. Besides the issues of a physical interpretation of such a ``state'' this causes 
further problems in evaluating interesting functions like the von Neumann entropy, the quantum Fisher information or an 
entanglement measure like the negativity as these functions are defined or meaningful only for valid, i.e., positive-semidefinite, quantum states.

For this reason, different methods have been introduced that mostly follow the paradigm that the reconstructed state 
$\hat \varrho = \arg \max\limits_{\varrho \geq 0}  T(\varrho|f)$ maximizes a target function  $T(\varrho|f)$ \emph{within} 
the set of \emph{valid} density operators. This target function thereby measures how well a density operator $\varrho$ agrees 
with the observed data $f$. Two common choices are maximum likelihood~(ML)~\cite{hradil97a} where $T_{\rm ML} = \sum_{r,s} f^s_{r} 
\log[P_\varrho^s(r)]$, and least squares~(LS)~\cite{james01a} where $T_{\rm LS} = -\sum_{r,s} [f^s_{r} - P_\varrho^s(r)]^2/
P_\varrho^s(r)$. We denote the respective solutions by $\hat \varrho_{\rm ML}$ and $\hat \varrho_{\rm LS}$. From these estimates one then easily 
computes any physical quantity of the observed state, like e.g. the fidelities $\hat{F}_{\rm ML} = \langle \psi | \hat{\varrho}_{\rm ML} 
| \psi \rangle$ and $\hat{F}_{\rm LS} = \langle \psi | \hat{\varrho}_{\rm LS} | \psi \rangle$ with respect to the target state $| \psi \rangle$.

%
%
\textit{Numerical simulations.}---To enable detailed analysis of the particular features of the respective state reconstruction 
algorithm and to exclude influence of systematic experimental errors we perform Monte Carlo simulations. 
For a chosen state $\varrho_0$ the following procedure is used:
{\it i}) Compute the single event probabilities $P^s_{\varrho_0}(r)$,
{\it ii}) toss a set of frequencies according to a multinomial distribution,
{\it iii}) reconstruct the state with either reconstruction method and compute the functions of interest,
{\it iv}) carry out steps {\it ii}) and {\it iii}) $500$ times. 
Note that the optimality of the maximizations for ML and LS in step (ii) is certified by convex optimization~\cite{moroder12a,cobook}.

Exemplarily, we first consider the four-qubit Greenberger-Horne-Zeilinger (GHZ) state $\ket{\GHZ_4}= (\ket{0000} 
+ \ket{1111})/\sqrt{2}$ mixed with white noise, i.e., $\varrho_0 = p\ket{\GHZ_4}\bra{\GHZ_4} + (1-p)\openone/16$ where $p$ is 
chosen such that the fidelity is $\braket{\GHZ_4| \varrho_0 |\GHZ_4} = 0.8$. This state is used to simulate the Pauli tomography 
scheme. Fig.~\ref{fig:sample1} shows an exemplary histogram of the resulting fidelities for $N_s=100$ measurement repetitions which 
is a typical value used for various multiqubit experiments. The fidelities obtained via LIN reconstruction fluctuate around the 
initial value ($\overline F_{\rm LIN}=0.799 \pm 0.012$). 
(The values given there are the mean and the standard deviation obtained from the 500 reconstructed states).
In stark contrast, both ML ($\overline F_{\rm ML}=0.788 \pm 0.010$) and 
even worse LS ($\overline F_{\rm LS}=0.749 \pm 0.010$) systematically underestimate the fidelity, i.e., are strongly biased.
Evidently, the fidelities of the reconstructed 
states differ by more than one standard deviation for ML and even more than five standard deviations for LS. The question arises 
how these systematic errors depend on the parameters of the simulation. Let us start by investigating the dependence on the number of 
repetitions $N_s$.
\begin{figure}
\includegraphics[width=\linewidth]{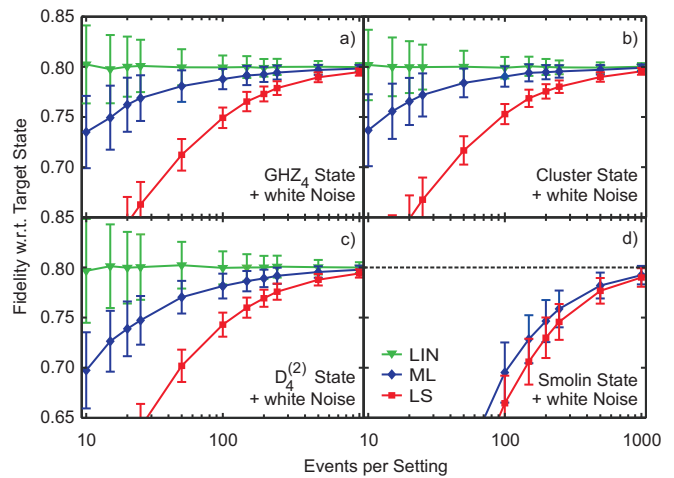}
\caption{(Color online) The performance of ML, LS, and LIN methods depending on the number of events $N_s$ per setting and for four different noisy initial 
states $\varrho_0$. Note that the fidelity can only be calculated linearly if the reference state is pure which is not the case 
for the Smolin state \cite{smolin01a}. Therefore only the curves for ML and LS are plotted for the Smolin state.}
\label{fig:sample2}
\end{figure}
Fig.~\ref{fig:sample2}a shows the mean and the standard deviations of histograms like the one shown in Fig.~\ref{fig:sample1}.
for different $N_s$. As expected, the systematic errors are more profound for low number of repetitions $N_s$ per setting 
$s$ and decrease with increasing $N_s$. Yet, even for $N_s = 500$, a number hardly used in multiqubit experiments, 
$\overline{F}_{\rm LS}$ still deviates by one standard deviation from the correct value. The effect is also by no means special 
for the GHZ state but was equally observed for other prominent four-party states, here also chosen with a true fidelity of 80\%, see 
Fig.~\ref{fig:sample2}b-\ref{fig:sample2}d and the Supplemental Material (SM).

The systematic deviations vary also with the number of qubits or the purity of the initial state. Fig.~\ref{fig:sample3}a shows 
the respective dependencies of the fidelity for $n$-qubit states $\varrho_0 = p \ket{\GHZ_n}\bra{\rm GHZ_n} + (1-p) \openone/2^n$ (for $N_s=100$).
Here, a significant increase of the bias with the number of qubits is observed especially for LS. Also when varying the purity or 
fidelity with the GHZ state, respectively, we observe  a remarkable deviation for ML and LS estimators (Fig.~\ref{fig:sample3}b). 
If the initial fidelity is very low, the effect is negligible, but large fidelity values suffer from stronger deviations, especially 
for LS. 

The commonly specified ``error bars'' used in QST quantify the statistical fluctuations of the estimate $\hat \varrho$. 
Starting either from the estimate $\hat \varrho_{\rm EST}$ (EST $\in$ $\{ \rm ML,\,LS \}$) or the observed data set $f$ this error is typically accessed by Monte Carlo sampling:
One repeatedly simulates data sets $f^{(i)}$ according to the state $\varrho_{\rm EST}$ or $f$ together with a representative noise model for the respective 
experiment and reconstructs the state $\hat \varrho(f^{(i)})$. From the resulting empirical distribution, one then 
reports the standard deviation (or a region including, say, $68\%$ of the simulated states) for the matrix elements or for quantities of interest \cite{bootstrapping}, see also SM\,3. 
%
%
However, the problem with such error bars is that they might be too small since they reflect only statistical fluctuations of the measured frequencies, but not the systematic error which easily can be larger~\cite{mood}. 

In summary, we observe systematic errors, which depend on the state reconstruction method and the strength of the statistical fluctuations 
of the count rates shown here as dependence on the number of repetitions of 
the experiment. Since the effect even depends on the unknown initial state any manual correction of the bias 
is unjustifiable. Let us emphasize that in most cases the initial value differs by more than the ``error bar'' determined via 
bootstrapping (cf. SM\,3).

\begin{figure}
\includegraphics[width=\linewidth]{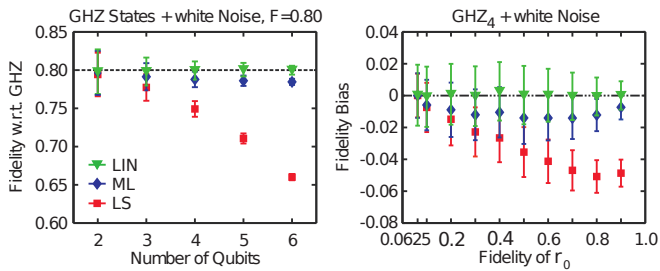}
\caption{(Color online) The behavior of ML, LS, and LIN depending on the number of qubits $n$ (left) and the fidelity of $\varrho_0$ (right).}
\label{fig:sample3}
\end{figure}

%
%
\textit{Biased and unbiased estimators.}---The systematic offset discussed above is well-known in the theory of point 
estimates~\cite{mood}. Expressed for QST, an estimator $\hat \varrho$ is called unbiased if its fluctuations are centered around 
the true mean, such that for its expectation value
\begin{equation}
\mathbb{E}_{\varrho_0}( \hat \varrho)\equiv \sum_{f} P_{\varrho_0}(f) \hat \varrho(f) = \varrho_0
\label{eq:unbiased_estimator}
\end{equation}
holds for all possible states $\varrho_0$ with $P_{\varrho_0}(f)$ the probability to observe the data $f$. An estimator that 
violates Eq.~\eqref{eq:unbiased_estimator} is called biased. Similar definitions hold for instance for fidelity estimators, 
$\mathbb{E}_{\varrho_0}(\hat F) = \braket{\psi|\varrho_0|\psi}\equiv F_0$. This terminology is motivated by the form of the mean 
squared error, which decomposes for example for the fidelity into
\begin{equation}
\mathbb{E}_{\varrho_0}[(\hat F - F_0 )^2] = \mathbb{V}_{\varrho_0} (\hat F) 
+ [ \mathbb{E}_{\varrho_0}(\hat F) - F_0]^2,
\label{eq:bias}
\end{equation}
where $\mathbb V(\hat F)\equiv\mathbb E(\hat F^2)-\mathbb E(\hat F)^2$ denotes the variance. Equation (\ref{eq:bias}) consists of 
two conceptually different parts. The first being a statistical term quantifying the fluctuations of the estimator $\hat F$ itself. 
The second, purely systematic term, is called {\it bias} and vanishes for unbiased estimators \cite{comment}. 
Note that, since the expectation values of the frequencies are the probabilities, ${\mathbbm E}_{\varrho_0}(f^s_{r})=
P^s_{\varrho_0}(r)$, and because $\hat\varrho_{\rm LIN}$ as given by Eq.~(\ref{eq:rholin}) is linear in $f^s_{r}$ the determination of a quantum state using LIN is unbiased. 
However, as shown below, for QST the bias is inherent to estimators constraint to giving only physical answers.

\begin{proposition*}
A reconstruction scheme for QST that always yields valid density operators is biased.
\end{proposition*}

\begin{proof}
For a  tomography experiment on the state $\ket{\psi_i}$ with finite measurement time there is a set of possible data 
${\cal S}_i=\{f_i | P_{\ket{\psi_i}}(f_i)>0\}$, with $P_{\ket{\psi_i}}(f_i)$ the probability to obtain data $f_i$ when 
observing state $\ket{\psi_i}$. 

Consider two pure non-orthogonal states $\ket{\psi_1}$ and $\ket{\psi_2}$ ($\langle\psi_1 | \psi_2 \rangle \neq 0$). For these 
two states there exists a non-empty set of data ${\cal S}_{12}=\{f' | P_{\ket{\psi_1}}(f') \cdot P_{\ket{\psi_2}}(f')>0\}=
{\cal S}_1 \cap {\cal S}_2$, which can occur for both states. 

Now let us assume that a reconstruction scheme $\hat \varrho$ provides a valid quantum state $\hat \varrho(f)$ for all possible 
outcomes $f$ and that Eq.~(\ref{eq:unbiased_estimator}) is satisfied for $\ket{\psi_1}$, i.e., $\sum_{{\cal S}_1} P_{\ket{\psi_1}}
(f_1) \hat \varrho(f_1) = \ket{\psi_1}\bra{\psi_1}$. This incoherent sum over all $\hat \varrho(f_1)$ 
can be equal to the pure state $\ket{\psi_1}\bra{\psi_1}$ only  for the (already pathological) case that $\hat \varrho(f_1)=
\ket{\psi_1}\bra{\psi_1}$ for all $f_1 \in {\cal S}_1$. This means that the outcome of the reconstruction is fixed for all $f_1$ 
including all data $f' \in {\cal S}_{12}$. As these data also occur for state $\ket{\psi_2}$ there exist $f_2 \in {\cal S}_{12}$ with 
$\hat \varrho(f_2) = \ket{\psi_1}\bra{\psi_1} \neq \ket{\psi_2}\bra{\psi_2}$. Thus, in Eq.~(\ref{eq:unbiased_estimator}), the sum 
over all reconstructed states now is an incoherent mixture of at least two pure states and the condition 
$\sum_{{\cal S}_2} P_{\ket{\psi_2}}(f_2) \hat \varrho(f_2) = \ket{\psi_2}\bra{\psi_2}$ is violated for $\ket{\psi_2}$. Hence, 
$\hat \varrho$ does not obey Eq.~(\ref{eq:unbiased_estimator}) for $\ket{\psi_2}$ and is therefore  biased~\cite{comment_proof}.
\end{proof}

This leaves us with the trade-off: Should one necessarily use an algorithm like ML or LS to obtain a valid quantum state but suffer 
from a bias, or should one use LIN which is unbiased but typically delivers an unphysical result?

%
%
\textit{Parameter estimation by linear evaluation.}---Here, we demonstrate that starting from $\hat \varrho_{\rm LIN}$ it is straightforward to provide a
valid, lower/upper bound and an easily computable 
confidence region for many quantities of interest. For that we exploit the fact that many relevant functions are either convex, like most 
entanglement measures or the quantum Fisher information, or concave, like the von Neumann entropy. We linearize these operators 
around some properly chosen state in order to obtain a reliable lower (upper) bound. Note that typically a lower bound on an 
entanglement measure is often suited for evaluating experimental states whereas an upper bound does not give much additional 
information.

Recall that a differentiable function $g(x)$ is convex if $g(x) \geq g(x') + \nabla g(x')^T (x-x')$ holds for all $x,x'$.
In our case we are interested in a function $g(x)=g[\varrho(x)]$ where $x$ is a variable to parametrize a quantum state
$\varrho$ in a linear way. From convexity it follows that it is possible to find an operator $L$, such that
\begin{equation}
 \tr(\varrho_{0} L) \leq g(\varrho_{0})  
\end{equation}
holds for all $\varrho_0$ (similarly an upper bound is obtained for concave functions). This operator can be determined from 
the derivatives of $g(x)$ with respect to $x$ at a suitable point $x'$. For cases where the derivative is hard to compute such an 
operator can also be obtained from the Legendre transformation \cite{guehne07z} or directly inferred from the definition 
of the function $g(x)$ \cite{vidal02a}. A detailed discussion is given in the SM\,5.

For this bound a confidence region, i.e., the error bars in the frequentistic approach, can be calculated. For example a
one-sided confidence region of level $\gamma$ can be described by a function $\hat C$ on the data $f$
such that $\Prob_{\varrho_{0}}[ \hat C  \leq g(\varrho_{0}) ] 
\geq \gamma$ holds for all $\varrho_0$~\cite{mood}. According to Hoeffding's tail inequality~\cite{tomamichel12a} 
and a given decomposition of $L = \sum l^s_{r} M^s_{r}$ into the measurement operators $M^s_{r}$ a confidence region then is 
given by
\begin{equation}
\label{eq:delta_hoeffding}
\hat C = \tr(\hat \varrho_{\rm LIN}L) - \sqrt{\frac{h^2 |\log(1-\gamma)|}{2 N_s}},
\end{equation}
where $h^2$ is given by $h^2 = \sum_s (l_{\rm max}^s - l_{\rm min}^s)^2$, and $l_{{{\rm max/min}}}^s$ denotes the respective extrema 
of $l^s_{r}$ over $r$ for each setting $s$. Although not being optimal, such error bars are easy to evaluate and valid without extra 
assumptions. Since we directly compute a confidence interval on $g(x)$ this is also generally a tighter error bar than those deduced 
from a ``smallest'' confidence region on density operators which tend to drastically overestimate the error (see SM\,4 for an example).

In the following we show how to use a linearized operator on the example of the bipartite negativity~\cite{vidal02a}. 
(For the quantum Fisher information~\cite{petz10a} and additional discussion see SM\,5.)
\begin{figure}
\includegraphics[width=\linewidth]{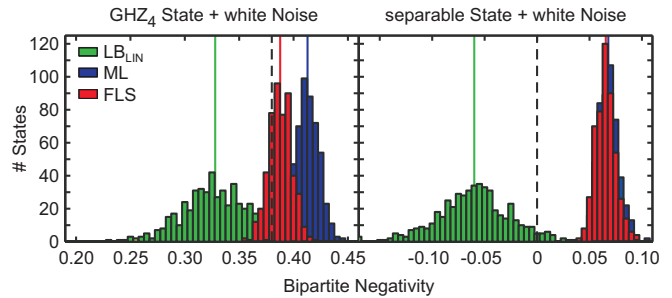}
\caption{(Color online) Lower bound ${\rm LB }_{\rm LIN}$ obtained by linearizing bipartite negativity   for a four-qubit product (left) and the GHZ state (right) both mixed with white noise resulting in $80\%$ fidelity. The ML and LS reconstruction leads to a systematic overestimation of the negativity, while the lower bound yields a 
valid estimate.}
\label{fig:negative}
\end{figure}
A lower bound on the negativity $N(\varrho_{AB})$ of a bipartite state $\varrho_{AB}$ is given by 
\begin{equation}
\label{eq:lb_negativity}
N(\varrho_{AB}) \geq \tr(\varrho_{AB}L)
\end{equation}
for any $L$ satisfying $\openone\geq L^{T_A} \geq 0$, where the superscript $T_A$ denotes partial transposition~\cite{peres96a} 
with respect to party $A$. The inequality~\eqref{eq:lb_negativity} is tight if $L$ is the projector on the negative eigenspace of 
$\varrho_{AB}^{T_A}$. Using this linear expression one can directly compute the lower bound on the negativity and by 
using Eq.~\eqref{eq:delta_hoeffding} the one-sided confidence region. 
Any choice of $L$ is in principle valid, however for a good performance $L$ should be chosen according to the experimental
situation. We assume, however, no prior knowledge and rather estimate $L$ by the projector on the negative eigenspace 
of $\hat \varrho_{\rm ML}^{T_A}$ deduced from an \emph{additional} 
tomography again with $N_s=100$ counts per setting. One can, of course, also start with an educated guess of $L$ motivated by the 
target state one wants to prepare. In any case, in order to apply Eq.~(\ref{eq:delta_hoeffding}) and to assure a linear evaluation of 
the data the operator $L$ must be chosen independently of the tomographic data \cite{moroder13a}.

Fig.~\ref{fig:negative} shows the distributions of the negativity between qubits $A=\{1,2\}$ and $B=\{3,4\}$ for the four-qubit 
GHZ state and for the separable four-qubit state $\ket{\psi_{\rm sep}} \propto  (\ket 0 + \ket +)^{\otimes 4}$, each mixed with white
noise such that the fidelity with the respective pure state is $80\%$.
In both cases we observe that ML and LS {\it overestimate} the amount of entanglement.
Even worse, if no entanglement is present, ML and LS clearly indicate entanglement. In 
contrast, the lower bound of the negativity, as given by Eq.~(\ref{eq:lb_negativity}), does not indicate false entanglement.

%
%
\textit{Conclusion.}---Any state reconstruction algorithm enforcing physicality of the result suffers from systematic 
deviations. We have shown that for the commonly used methods and the typical measurement schemes this bias 
is significant for data sizes typical in current experiments. It leads to systematically wrong statements about derived quantities 
like the fidelity or the negativity which can lead to erroneous conclusions particularly for the presence of entanglement.
Equivalent statements can be inferred for process tomography.

We have demonstrated that the simple method of linear inversion can be used to overcome these problems in many cases. 
Expectation values being linear in $\varrho$ do not exhibit a bias at all even if $\hat\varrho_{\rm LIN}$ is not physical in 
the overwhelming number of cases. A linearization of convex (concave) nonlinear physical quantities yields meaningful lower (upper) bounds together
with easy to calculate confidence intervals restoring the trust in quantum state and process tomography.

\begin{acknowledgments}
We like to thank D.\ Gross, P.\ Drummond and M.\ Raymer for stimulating discussions. This work has been supported by the 
EU (Marie Curie CIG 293993/ENFOQI and ERC QOLAPS), the BMBF (Chist-Era project QUASAR) and QCCC of the Elite Network of Bavaria.
\end{acknowledgments}


\begin{thebibliography}{49}
\bibitem[{com({\natexlab{a}})}]{comment}
\bibinfo{note}{It is well-known that estimators can be biased for a finite
  number of samples or for particular additional conditions and constraints. An
  example is the estimate of the variance $\sigma^2$ of a Gaussian distribution
  with known mean $\mu$. The ML estimate from $N$ samples $x_i$ gives $\hat
  \sigma^2_{\rm ML}(x_i) = \frac{1}{N} \sum_i (x_i -\mu)^2$, while an unbiased
  estimator is $\hat \sigma^2 (x_i) = \frac{1}{N-1} \sum_i (x_i -\mu)^2$}.
\bibitem[{com({\natexlab{b}})}]{comment_proof}
\bibinfo{note}{This proof, given here for pure states, can be generalized for
  mixed states $\varrho_1$ and $\varrho_2$ with the ranges of the two states
  not sharing any vector.}
\end{thebibliography}


%
%
\appendix
\section{Supplemental Material}
\section{SM\,1: Quantum state reconstruction using linear inversion}

In~\cite{james01a} it is explained how to obtain the estimate $\hat\varrho_{\rm LIN}$ for an $n$-qubit state from
the observed frequencies of a complete set of projection measurements, i.e. $4^n$ results. Yet, the scheme described there
is more general and can be used for any (over)complete set of projection measurements. 

In the standard Pauli basis $\{\sigma_0, \sigma_x, \sigma_y \sigma_z \}$ the density
matrix of the state $\varrho$ is given by
\begin{equation}
  \varrho = \frac{1}{2^n}\sum_{\mu} T_{\mu} \Gamma_{\mu}
\end{equation}
where $\mu = 1...4^n$ enumerates all possible $n$-fold tensor products of Pauli matrices $\Gamma_1 = \sigma_0 \otimes \sigma_0 \otimes ... \otimes \sigma_0$,
$\Gamma_2 = \sigma_0 \otimes \sigma_0 \otimes ... \otimes \sigma_x$, etc.
and with correlations $T_{\mu} = \tr(\varrho \Gamma_{\mu})$.
To simplify our notation we will use the following mapping for a setting $s$ with a respective outcome $r$:
$(r,s)\longrightarrow\nu = 2^{n(s-1)} + r - 1$, hence for the projectors, $M_r^s \longrightarrow M_{\nu}$, and for the $A_r^s \longrightarrow A_{\nu}$, etc.
Then the probabilities to observe a result $r$ for setting $s$, or $\nu$ respectively, are given by 
\begin{equation}
  P_{\nu} = \tr(\varrho M_{\nu}) =  \frac{1}{2^n} \sum_{\mu} \tr(M_{\nu} \Gamma_{\mu}) T_{\mu}.
  \label{eq:Pnu}
\end{equation}
Introducing the matrix $\hat B$ with elements
\begin{equation}
 B_{\nu,\mu} = \frac{1}{2^n} \tr(M_{\nu} \Gamma_{\mu})
\end{equation}
Eq. (\ref{eq:Pnu}) simplifies to
\begin{equation}
 \vec{P} = \hat B \vec{T}.
 \label{eq:Pvec}
\end{equation}
Inverting Eq. (\ref{eq:Pvec}), the correlations can be obtained from the probabilities $P_{\nu}$, i.e., $T_{\mu} = \sum_{\nu} (\hat B^{-1})_{\mu,\nu} P_{\nu}$. 
Note that this is possible for any set of measurement operators. In case of a tomographically overcomplete set, i.e. $\nu > \mu$
the inverse $\hat B^{-1}$ has to be replaced by the pseudo inverse $\hat B^{-1} \longrightarrow B^{+} = (B^{\dagger} B)^{-1} B^{\dagger}$.
Reinserting $T_{\mu}$ one obtains 
\begin{equation}
 \varrho =  \frac{1}{2^n} \sum_{\nu,\mu} (\hat B^{-1})_{\mu,\nu} \Gamma_{\mu} P_{\nu}.
\end{equation}
For finite data sets, the $P_{\nu}$ are replaced by the frequencies $f_{\nu}$ and with
\begin{equation}
A_{\nu} = \frac{1}{2^n} \sum_{\mu} (\hat B^{-1})_{\mu,\nu} \Gamma_{\mu}
\end{equation}
Eq.~(1) is obtained.

\section{SM\,2: Bias for other prominent states}

The occurrence of a bias for fidelity estimation based on ML and LS state reconstruction is by no means a special feature of the 
GHZ state. In Fig.~\ref{fig:sample6} we show some further examples of the corresponding dependencies of the bias on the number of 
measurements per setting $N_s$ for the $W$ and the fully separable state $\ket{\psi} \propto (\ket{0}+\ket{+})^{\otimes 4}$. For 
all these pure states we assume that they are mixed with white noise for an overall initial fidelity of 80\%, so that the states
are not at the border of the state space.

\begin{figure}[!h]
\includegraphics[width=\linewidth]{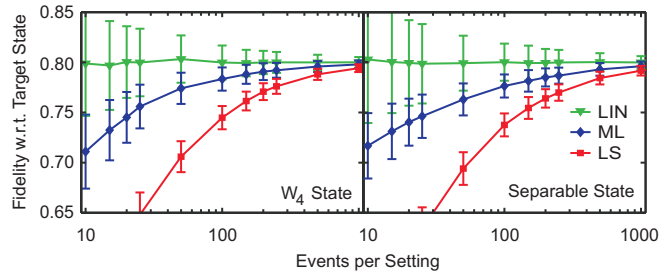}
\caption{The behavior of ML, LS and LIN depending on the number of events $N_s$ per setting for different noisy initial states 
$\varrho_0$.}
\label{fig:sample6}
\end{figure}

Furthermore we observed that the fidelity values as inferred via LS are systematically lower than those obtained using ML, see 
Fig.~\ref{fig:hist_low}.
\begin{figure}[!h]
\includegraphics[width=\linewidth]{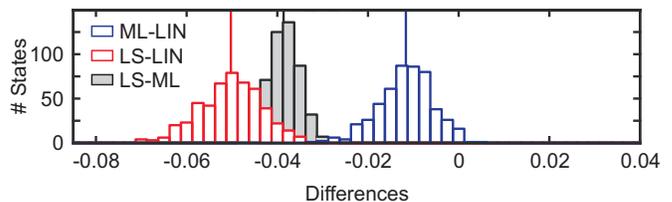}
\caption{Here we show the differences of the respective fidelity estimates evaluated for each single simulated tomography experiment as shown in Fig.~1
of the main text. 
It shows that the respective ML or LS estimate, with one rare exception, is always lower than the LIN estimate. 
Comparing ML and LS (gray) shows that not only on average but also for every single data set LS delivers a smaller fidelity value than ML.
}
\label{fig:hist_low} 
\end{figure}

\section{SM\,3: Bootstrapping}

As already mentioned in the main text, in many publications where QST is performed the standard error bar is calculated by bootstrapping
based on Monte Carlo methods. One can here
distinguish between parametric bootstrapping, where $f^{(i)}$ are sampled according to $\hat P^s(r) = \tr(\hat \varrho(f_{\rm obs}) 
M^s_{r})$, and non-parametric bootstrapping, where $\hat P^s(r) = f_{\rm obs}$ is used instead. 

We consider again the four-qubit GHZ state of $80\%$ fidelity and $N_s=100$. 
Interpreting the simulations of Fig.~1 as Monte-Carlo simulations from the parametric bootstrap with $\hat P^s(r) = \tr(\varrho_0 M^s_{r})$ we have 
already seen that ML and LS yield fidelity estimates below the actual value.
If one uses now one of these data sets $f_{\rm obs}$ 
as a seed to generate new samples $f^{(i)}$ the fidelity decreases further. As shown in Fig.~\ref{fig:sample4} this happens in 
particular for parametric bootstrapping ($0.777 \pm 0.011$ for ML and $0.700 \pm 0.012$ for LS) while non-parametric bootstrapping 
($0.780 \pm 0.011$ for ML and $0.714 \pm 0.012$ for LS) weakens this effect. 
However, in this context, one is interested in fact in the standard deviation of the simulated distribution. In our 
simulations it is somewhat smaller than the distribution of linearly evaluated fidelities. This means, the biasedness of ML and LS 
methods leads to a false estimate of the error, too.

\begin{figure}[!b]
\includegraphics[width=\linewidth]{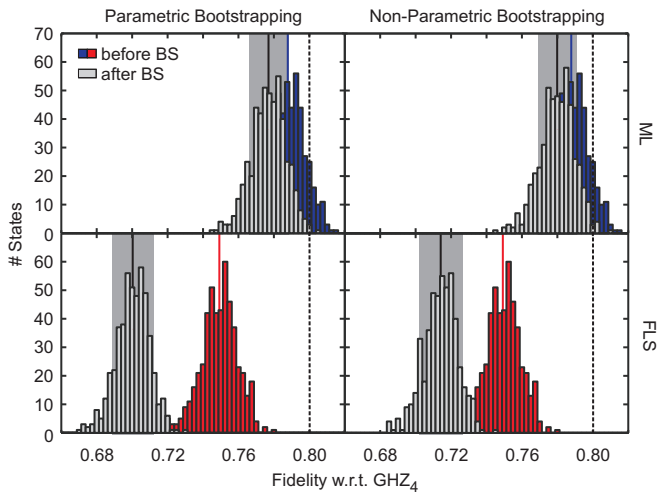}
\caption{Error bar computation for the fidelity of the four-qubit GHZ state via Monte-Carlo simulation using either parametric or 
non-parametric bootstrapping with the data from Fig.~1. For each of these $500$ observations $f_{\rm obs}$, $100$ 
new data sets $f^{(i)}$ were generated and reconstructed in order to deduce the mean and standard deviation as an error bar for the 
fidelity. The histograms denoted by ``after BS'' show the distributions of these means together with an averaged error bar given by 
the gray shaded areas. The initial values for the fidelities are described by the dashed lines.}
\label{fig:sample4}
\end{figure}

\section{SM\,4: Confidence regions for states vs.~scalar quantities}
 
Let us now comment on confidence regions (CR) for density operators and CR on parameter functions $Q$. Having a (tractable) method 
to compute CR for states $\hat C_{\varrho}(f)$~\cite{christandl11a}, one could think that this region of states also 
provides good CR for the parameter functions $Q$, if one manages to evaluate the minimal and maximal values of $Q(\varrho)$ for all 
$\varrho \in \hat C_\varrho(f)$. However, such CR are typically much worse than CR evaluated for $Q$ 
directly, the reason being the large freedom in how to build up a CR. Let us give 
the following illustrative example, see also Fig.~\ref{fig:cr}:

\begin{figure}
\includegraphics[width=\linewidth]{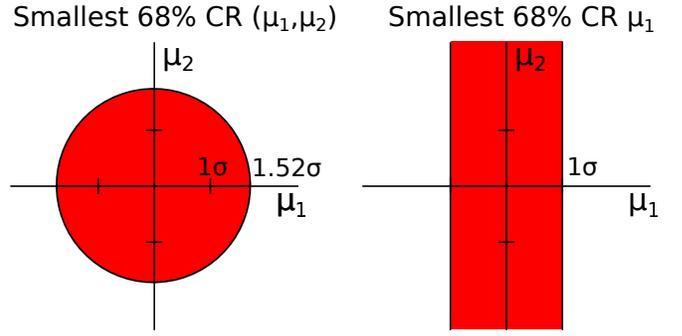}
\caption{Which confidence region is the smallest? If one is interested in both mean values $\vec \mu = (\mu_1,\mu_2)$ then clearly 
the left one represents the smallest one, but if $Q(\vec \mu)=\mu_1$ is chosen, then the right one is much better than the projected 
left one.}
\label{fig:cr}
\end{figure}

Let us consider the task to obtain a CR for the two mean values $\vec \mu = (\mu_1, \mu_2)$ of two independent Gaussian experiments, 
where the first $N$ samples $x_i$ are drawn from $\mathcal{N}(\mu_1, \sigma^2)$ while the remaining $N$ instances $y_i$ originate 
from $\mathcal{N}(\mu_2, \sigma^2)$, both with the same known variances. If one is interested in an $68\%$ CR for both mean values 
$\vec \mu$ then both possible recipes 
\begin{align}
\hat C^{(1)} &= \{ \vec \mu \: : \: \| \vec \mu - (\bar x, \bar y) \| \leq 1.52 \sigma/\sqrt{N}\}, \\
\hat C^{(2)} &= [\bar x - \sigma/\sqrt{N}, \bar x + \sigma/\sqrt{N}] \times (-\infty, \infty)
\end{align}
with $\bar x = \frac{1}{N}\sum_i x_i$ and similar for $\bar y$ are valid $68\%$ CR. However, while $\hat C^{(1)}$ 
yields the smallest area for the CR, it gives a much larger confidence region for $Q(\vec \mu) = \mu_1$ than if we would directly use 
$\hat C^{(2)}$, which in fact is the smallest one for $\mu_1$. Note that this effect increases roughly with $\sqrt{\rm dim}$ if one 
adds further parameters in the considered Gaussian example. Therefore we see that ``error bars'' associated with CR on the density 
operator are not the best choice if one is interested only on a few key properties of the 
state.
 
\section{SM\,5: Bounds on convex/concave functions}

As mentioned in the main text, one can directly bound convex (or concave) functions $g(x)$ by linear ones using an operator $L$
\begin{equation}
\tr(\varrho_0 L) \leq g(\varrho_0).
\label{eq:operator}
\end{equation}
Here, we want to explain in detail how the operator $L$ can be determined from the derivatives  of $g[\varrho(x)]$. Therefore, we 
parametrize the density operator $\varrho(x)=\openone/\dim + \sum_i x_i S_i$ via an orthonormal basis $S_i$ of Hermitian traceless 
operators. A possible choice for the $S_i$ are all normalized traceless tensor products of the Pauli matrices and 
the identity. Since we employ an affine parametrization, the function $g(x)=g[\varrho(x)]$ is convex. Direct calculation shows that 
choosing the operator $L[\varrho(x')]=l_0\openone+\sum_i l_i S_i$ as
\begin{eqnarray}
l_0 &=& g[\varrho_{\rm guess}(x')]-\sum_i x'_i \frac{\partial}{\partial x_{i}}g[\varrho_{\rm guess}(x')]\\
l_i &=& \frac{\partial}{\partial x_i} g[\varrho_{\rm guess}(x')]
\end{eqnarray}
gives due to the convexity condition $g(x) \geq g(x') + \nabla g(x')^T (x-x')$
a lower bound as in Eq.~(\ref{eq:operator}). Here, $L[\varrho(x')]$ is computed on a ``guess'' $x'$, 
i.e., $\varrho_{\rm guess}(x')$ of the true state $\varrho_0$. Recall that while the guess $\varrho_{\rm guess}$ must be a valid 
state the lower bound $\tr(\varrho_0 L)$ is well-defined also for nonphysical density operators.

\begin{figure}
\includegraphics[width=\linewidth]{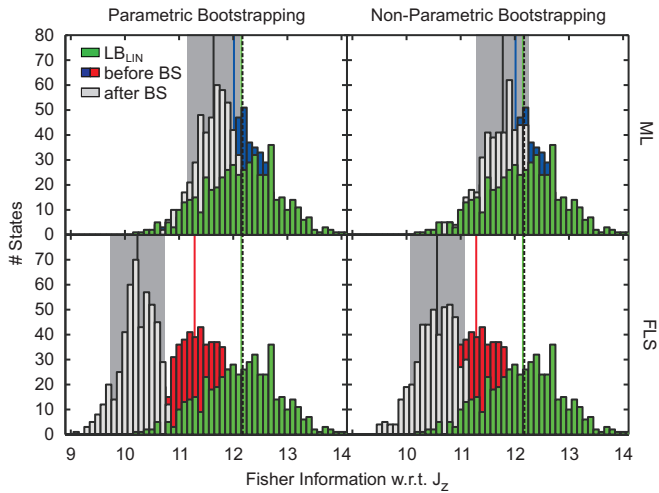}
\caption{Full analysis of a Pauli QST scheme with $N_s=100$ on four qubits in order to deduce the quantum Fisher information with 
respect to $H=J_z$. As the true underlying state we assume again a noisy four-party GHZ state. We observe that the quantum Fisher 
information is underestimated from both ML and LS, while the lower bound deduced from LIN is fine. }
\label{fig:sample5}
\end{figure}

As an example how to apply this linearization, let us consider the quantum Fisher information $f(x) = F(\varrho,H)$, which measures 
the suitability of a state $\varrho$ to determine the parameter $\theta$ in an evolution $U(\theta,H)=e^{-i\theta H}$. More 
explicitly the formulae are given by
\begin{align}
f(x) &= 2\sum_{jk}\frac{(\lambda_j-\lambda_k)^2}{\lambda_j+\lambda_k} 
H_{jk}H_{kj}, \\
\frac{\partial}{\partial x_i} f(x) & = 4 \sum_{jkl} \frac{\lambda_j \lambda_k + \lambda_j \lambda_l + \lambda_k \lambda_l -3 
\lambda_j^2}{(\lambda_j+\lambda_k)(\lambda_j+\lambda_l)} H_{jk}S_{i, kl}H_{lj}
\end{align} 
where $\{\lambda_i, \ket{\psi_i}\}$ denotes the eigenspectrum of $\varrho(x)$, $H_{jk}=\braket{\psi_j|H|\psi_k}$ and 
$S_{i,kl}=\braket{\psi_k|S_i|\psi_l}$. In order to compute the derivative of the Fisher information one can employ the alternative 
form, as given for instance in Ref.~\cite{petz10a},
\begin{eqnarray}
F(\varrho, H)&=&\tr[(H\varrho^2 + \varrho^2 H -2\varrho H \varrho) J^{-1}_{\varrho}(H)],\\
J^{-1}_{\varrho}(H)&=&\int_{0}^\infty dt\:\: e^{-t/2 \varrho} H e^{-t/2 \varrho}.
\end{eqnarray}
such that the derivative can be computed via the help of matrix derivatives~\cite{horn91a}.

Now let us imagine that we want to determine the quantum Fisher information of a four-qubit state with respect to $H=J_z$, while our 
true underlying state $\varrho_0$ is once more the noisy GHZ state of $80\%$ fidelity. Figure~\ref{fig:sample5} shows the full 
simulation of a Pauli tomography experiment with $N_s=100$ together with the standard error analysis using parametric or 
non-parametric bootstrapping. As with the other examples, we observe a systematic discrepancy between the results of standard QST 
tools and the true value. In this case, though the quantum Fisher information is typically larger for stronger entangled states, ML or LS 
underestimate the true capabilities of the state. However, if we use the described method for LIN (with an in this case optimized 
operator $L$) the lower bound via LIN is fine.

For completeness, we also give the respective derivatives for further convex functions of interest like the purity 
$g(x) = \tr(\varrho^2)$ 
\begin{equation}
\frac{\partial}{\partial x_i} g(x) = 2 \tr[S_i \varrho(x)]
\end{equation}
and correspondingly for the von Neumann entropy $g(x) = -\tr(\varrho \log \varrho)$ 
\begin{equation}
\frac{\partial}{\partial x_i} g(x) = -\tr[S_i(\log \varrho(x) -\mathbbm{1})].
\end{equation}

\end{document}